\documentclass[a4paper,10pt]{article}
\usepackage[small]{titlesec}
\usepackage[utf8]{inputenc}
\usepackage[T1]{fontenc}
\usepackage{lmodern}
\usepackage[leqno]{amsmath}
\usepackage{amsfonts,amsthm,amssymb,mathtools,dsfont,bbm}
\usepackage[dvipsnames]{xcolor}
\usepackage[left=2.5cm,right=2.5cm,top=2.5cm,bottom=2.5cm]{geometry}
\usepackage{tikz,amssymb}
\usetikzlibrary{arrows.meta}
\numberwithin{equation}{section}
\makeatletter
\newcommand{\reqnomode}{\tagsleft@false\let\veqno\@@eqno}
\makeatother
\newcommand\myshade{85}
\colorlet{mylinkcolor}{violet}
\colorlet{mycitecolor}{YellowOrange}
\colorlet{myurlcolor}{Aquamarine}
\usepackage[backend=biber,style=alphabetic,sorting=nyt,
 giveninits=true,maxbibnames=99,isbn=false]{biblatex}
\DeclareNameAlias{sortname}{family-given}
\AtBeginBibliography{\setlength{\emergencystretch}{1em}}
\DeclareFieldFormat[article]{title}{#1}
\DeclareFieldFormat{journaltitle}{#1}
\DeclareFieldFormat{volume}{\mkbibbold{#1}}
\renewbibmacro{in:}{}
\AtEveryBibitem{\iffieldundef{doi}{}{\clearfield{eprint}\clearfield{url}}}
\usepackage[unicode=true,pdfusetitle,colorlinks=true,
 linkcolor=mylinkcolor!\myshade!black,
 citecolor=mycitecolor!\myshade!black,
 urlcolor=myurlcolor!\myshade!black]{hyperref}
\usepackage[nameinlink]{cleveref}
\newtheorem{thm}{Theorem}[section]
\newtheorem{lem}[thm]{Lemma}
\newtheorem{prop}[thm]{Proposition}

\newtheorem{conj}[thm]{Conjecture}
\theoremstyle{remark}
\crefname{thm}{Theorem}{Theorems}
\crefname{lem}{Lemma}{Lemmas}
\crefname{prop}{Proposition}{Propositions}
\crefname{cor}{Corollary}{Corollaries}
\crefname{conj}{Conjecture}{Conjectures}
\crefname{rem}{Remark}{Remarks}
\crefname{section}{Section}{Sections}
\crefformat{equation}{(#2#1#3)}
\newcommand{\ee}{\rm e}
\newcommand{\eq}[1]{\begin{align*}#1\end{align*}}
\newcommand{\eql}[1]{\begin{align}#1\end{align}}
\newcommand{\ZZ}{\mathbb Z}
\newcommand{\EE}{\mathbb E}

\newcommand{\PP}{\mathbb P}
\newcommand{\calL}{\mathcal L}
\newcommand{\Id}{\mathds{1}}
\DeclareMathOperator{\Var}{\mathbbm{Var}}
\newcommand{\br}[1]{\left(#1\right)}

\title{Comparing the Lee--Yang gap with the mass gap in the sub-critical planar Ising model}
\author{
    \href{mailto:nborgnia@gmail.com}{Noam Borgnia}\\
	{\footnotesize Technische Universität München}\\\href{mailto:jc1220@math.princeton.edu}{Jui-Hui Chung}\\
	{\footnotesize Program in Applied and Computational Mathematics, Princeton University }\\
\href{mailto:jacobshapiro@princeton.edu}{Jacob Shapiro}\\
{\footnotesize Department of Mathematics, Princeton University}\\
{\footnotesize and Faculty of Mathematics, University of Vienna}\\[0.5em]
}
\date{September 17 2026}

\begin{document}
\reqnomode
\maketitle
\begin{abstract}
We prove that the subcritical mass gap is comparable to the
Lee--Yang gap over free-boundary square boxes raised to the
power $8/15$. Random-current moment bounds, a correlation inequality and a finite block expansion give the upper bound. The lower bound follows from the Lee--Yang product, the susceptibility and the critical magnetization. We also obtain uniform factorial bounds on the field derivatives of the pressure.
\end{abstract}

\section{Statement and notation}\label{sec:statement}
The Combes--Thomas estimate of \cite[Section 10.3]{AW} relates distance from the spectrum to spatial decay of the Greens function: for a Schr\"odinger operator $H=-\Delta+V(X)$ on $\ell^2(\ZZ^d)$, if $0<\delta=\operatorname{dist}(E,\operatorname{spec}H)\le1$, then $|(H-E\Id)^{-1}(x,y)|\le C\delta^{-1}e^{-c\delta|x-y|_1}$ for all $x,y\in\ZZ^d$. For real energies below the spectrum, $E<\inf\sigma(H)$,
the estimate improves to
\eq{
 |(H-E\Id)^{-1}(x,y)|
 \le C\delta^{-1}e^{-c\sqrt{\delta}|x-y|_1}\qquad(x,y\in\ZZ^d).
}
Thus, writing $m_H(E)$ for the optimal uniform exponential
decay rate of the Green's function, we obtain
\eql{
 m_H(E)\gtrsim\sqrt{\delta}\,.
} Of course we always have 
\eql{
 m_H(E)\leq \sqrt{\delta+\sup V - \inf\sigma (H)}
} via Feynman-Kac.

It is tempting to ask whether an analogous relation may hold for the Ising model in its sub-critical, disordered phase, because one may associate the Greens function with the two-point function and the distance to the spectrum with \emph{the Lee-Yang gap}: the magnitude of the closest Lee-Yang zero to the origin. The association of these two quantities is not outlandish by itself. Penrose and Lebowitz \cite{PL} prove that correlations decay exponentially at positive real fugacities in a connected zero-free region containing zero, provided this region is uniform for sufficiently large periodic rectangles. At the fugacity corresponding to zero magnetic field, a strictly positive Lee-Yang gap supplies their zero-free hypothesis, while exponential decay gives a positive mass gap. Their result thus provides a qualitative precedent, but not a quantitative comparison of the two gaps. Furthermore, in one dimension, where the critical point is at zero temperature, $\beta=\infty$, Lee and Yang's calculation \cite[Section V.E, formulas (43)--(44)]{LY} places the zero edge at $H=\mathrm{i}\arcsin(e^{-2K})$, with $K$ the dimensionless coupling and $H$ the dimensionless field. Here, as below, the field enters as $\beta h$, so $K=\beta$ and $H=\beta h$. Combining this with the routine zero-field transfer-matrix calculation gives
\eq{
 m_\beta^{(\mathrm{1D})}=-\log\tanh\beta,
 \qquad \calL_\beta^{(\mathrm{1D})}=\frac1\beta\arcsin(e^{-2\beta});
} Consequently,
\eq{
 m_\beta^{(\mathrm{1D})}\sim2e^{-2\beta},
 \qquad \calL_\beta^{(\mathrm{1D})}\sim\frac{e^{-2\beta}}\beta,
 \qquad m_\beta^{(\mathrm{1D})}\sim2\beta\calL_\beta^{(\mathrm{1D})}
 \quad(\beta\to\infty).
}
Thus the two gaps are asymptotically proportional when the Lee--Yang gap is measured in the dimensionless field $\beta h$. In this note we compare these two gaps in the planar, $\ZZ^2$, Ising model as $\beta\uparrow\beta_c=\frac12\log(1+\sqrt2)<\infty$.

For a nonempty finite $\Lambda\subset\ZZ^2$, let $E(\Lambda)$ be its nearest-neighbour edges and set
\eql{\label{eq:partition}
 Z_{\beta,\Lambda}(h)&=\sum_{\sigma\in\{-1,1\}^\Lambda}
 \exp\left(\beta\sum_{\{x,y\}\in E(\Lambda)}\sigma_x\sigma_y+\beta h M_\Lambda\right),
 &M_\Lambda&=\sum_{x\in\Lambda}\sigma_x.
}
All finite domains have free boundary conditions. Write $\EE_{\beta,\Lambda,h}$ for the corresponding expectation, omitting $h$ at zero field and omitting $\Lambda$ when $\Lambda=\ZZ^2$. For $L\in\{0,1,2,\ldots\}$, put $\Lambda_L=[-L,L]^2\cap\ZZ^2$, $M_L=M_{\Lambda_L}$ and $Z_{\beta,L}=Z_{\beta,\Lambda_L}$. Onsager's critical-point relation \cite[formula (121)]{Ons} is $\sinh(2K_1)\sinh(2K_2)=1$, where $K_1,K_2$ are the dimensionless couplings in the two lattice directions. The isotropic choice $K_1=K_2=\beta$ gives the critical inverse temperature $\beta_c=\frac12\log(1+\sqrt2)$. For $0<\beta<\beta_c$, \emph{the axial mass gap} is
\eql{\label{eq:mass}
 m_\beta=-\lim_{n\to\infty}\frac1n\log\EE_\beta[\sigma_0\sigma_{(n,0)}].
} The transfer-matrix calculation in \cite[formula (3.36a)]{SML} gives the zero-momentum excitation energy $\epsilon_0=2(K_1^*-K_2)$ in the disordered phase, with dual coupling $K_1^*=-\frac12\log\tanh K_1$. For $K_1=K_2=\beta$, this is the axial decay rate $m_\beta$; expanding at $\beta_c$ gives
\eql{\label{eq:rigorous mass gap}
m_\beta = -2\beta-\log\tanh\beta = 4\left(\beta_c-\beta\right)+O(|\beta_c-\beta|^2)\,.
}

Define, in the physical field convention of \cref{eq:partition}, the \emph{Lee-Yang gap} as
\eql{\label{eq:gap}
 \calL_{\beta,\Lambda}=\inf\{s>0:Z_{\beta,\Lambda}(\mathrm{i} s)=0\},
 \qquad \calL_\beta=\inf_{L\ge0}\calL_{\beta,\Lambda_L}.
}
The Lee--Yang theorem \cite{LY} places the zeros of the ferromagnetic Ising partition polynomial on the fugacity circle $|z|=1$. Since $z=e^{-2\beta h}$ in the convention of \cref{eq:partition}, this is precisely $\operatorname{Re}h=0$, so all zeros lie on the imaginary axis.
The results of Penrose and Lebowitz \cite[Theorems 1--2]{PL} give a transfer-matrix gap and exponential decay at positive real fugacities in a connected zero-free region containing zero, provided that region persists for sufficiently large periodic rectangles. At the fugacity corresponding to zero magnetic field, the decay estimate becomes $\EE_\beta[\sigma_0\sigma_x]\le Ce^{-c|x|}$ with $c>0$, and hence $m_\beta>0$, under their uniform Lee--Yang-gap hypothesis. There are (to our knowledge) no known rigorous results on $\calL_\beta$ analogous to \cref{eq:rigorous mass gap}, only a scaling \emph{prediction}. In the continuum description of Fonseca and Zamolodchikov \cite{FZ}, the Yang--Lee edge occurs at a fixed imaginary value of $h_{\rm cont}/|m_{\rm cont}|^{15/8}$, where $h_{\rm cont}$ is the magnetic coupling and $m_{\rm cont}$ the thermal mass parameter. These correspond, to leading order and up to nonzero normalization constants, to $\beta h$ and $\beta-\beta_c$, respectively. In our field convention, the prediction is therefore
\eql{
\calL_\beta \sim A \left(\beta_c-\beta\right)^{15/8}
} for some $A\in(0,\infty)$\,.

The following result gives the near-critical power relating the gaps for free square boxes.

\begin{thm}\label{thm:main}
There are $c,C>0$ and $0<\beta_0<\beta_c$ such that,
\eql{\label{eq:main}
 c\calL_\beta^{8/15}\le m_\beta\le C\calL_\beta^{8/15}
 \qquad(\beta\in(\beta_0,\beta_c))\,.
} 

Moreover, for every $L\ge0$ and $k\ge1$,
\eql{\label{eq:cumulant-main}
 \frac1{|\Lambda_L|}\left|\left.\partial_h^{2k}\log Z_{\beta,L}(h)\right|_{h=0}\right|
 \le (2k)!C^{2k}m_\beta^{\,2-15k/4}.
}
\end{thm}
\bigskip
\noindent We leave open whether the limit
\eql{
 \lim_{\beta\to\beta_c^-}
 \frac{\calL_\beta^{8/15}}{m_\beta}
}
exists. By \cref{thm:main}, any such limit would necessarily
belong to $(0,\infty)$.

\paragraph{Where does $15/8$ come from?} Wu's critical asymptotic \cite{Wu} gives the decay $\EE_{\beta_c}[\sigma_0\sigma_x]\asymp|x|^{-1/4}$ for the zero-field square-lattice model. This is the input that suggests the exponent $15/8$. In a block of diameter $R$, summing this correlation over pairs gives a magnetization variance of order $R^2R^{2-1/4}=R^{15/4}$, hence fluctuations of order $R^{15/8}$. The field term $\beta hM$ becomes significant when $hR^{15/8}$ is of order one. Cutting off at the correlation length $R\asymp m_\beta^{-1}$ predicts $h\asymp m_\beta^{15/8}$. The proof implements this heuristic by controlling block magnetization moments with random currents and separating blocks at a sufficiently large multiple of the correlation length.

\paragraph{Higher dimensions.} The random-current moment bound and the block expansion below are dimension independent. For the nearest-neighbour Ising model on $\ZZ^d$, $d\ge2$, a critical bound $\EE_{\beta_c}[\sigma_0\sigma_x]\le C(1+|x|)^{-a}$, with $0<a<d$, gives $\chi_\beta:=\sum_x\EE_\beta[\sigma_0\sigma_x]\le C m_\beta^{-(d-a)}$ by the same spectral argument. If, in addition, the separately wired block-to-shell correlation can be made uniformly small at block size $\ell\asymp m_\beta^{-1}$, the expansion yields $\calL_\beta\ge c/(\beta\ell^{d/2}\sqrt{\chi_\beta})\ge c m_\beta^{d-a/2}$. In the customary notation $a=d-2+\eta$, the exponent is $(d+2-\eta)/2$, recovering $15/8$ when $d=2$. The additional mixing estimate is supplied here by planar crossing bounds; it does not follow from the critical two-point bound alone. A matching upper bound on $\calL_\beta$ additionally requires a near-critical susceptibility lower bound $\chi_\beta\ge c m_\beta^{-(d-a)}$ and a critical magnetization upper bound $\EE_{\beta_c,h}[\sigma_0]\le C h^{a/(2d-a)}$.

\paragraph{What about RFIM?} For random fields, spatial decay and a complex-field gap must be distinguished. Given real fields $\eta=(\eta_x)$, replace $\beta hM_\Lambda$ in \cref{eq:partition} by $\beta\sum_{x\in\Lambda}(h+\eta_x)\sigma_x$, and denote the resulting partition function by $Z^\eta_{\beta,\Lambda}(h)$. Set
\eq{
 \calL^\eta_{\beta,\Lambda}=\inf\{|h|:Z^\eta_{\beta,\Lambda}(h)=0\},
 \qquad \calL^\eta_\beta=\inf_{L\ge0}\calL^\eta_{\beta,\Lambda_L}.
}
These zeros need not lie on the imaginary axis. Two results of Ding, Song and Sun \cite[Corollaries 1.3 and 1.8]{DSS} are useful here. For arbitrary real dimensionless fields $g$, they prove the correlation comparison $\langle\sigma_x\sigma_y\rangle_g-\langle\sigma_x\rangle_g\langle\sigma_y\rangle_g\le\langle\sigma_x\sigma_y\rangle_0$. For $\beta<\beta_c$, they also establish the strong spatial mixing estimate $\|\mu^{\tau,g}_{\Lambda}|_A-\mu^{\tau^y,g}_{\Lambda}|_A\|_{\rm TV}\le Ce^{-c\operatorname{dist}(A,y)}$, uniformly in $g$ and finite domains. Here $\langle\cdot\rangle_g$ and $\mu^{\tau,g}_\Lambda$ denote the Ising expectation and Gibbs law with field term $\sum_xg_x\sigma_x$, $\tau^y$ differs from $\tau$ only at the boundary spin $y$, and $\mu|_A$ is the marginal on $A$. Taking $g_x=\beta\eta_x$ and nearest-neighbour coupling $\beta$ gives both estimates for the present random-field family. The finite-size complete-analyticity criterion of Dobrushin and Shlosman \cite[Sections 2 and 4]{DS} turns sufficiently small boundary influence on a fixed block into a complex neighbourhood, uniform in the volume and boundary conditions, in which $Z_\Lambda(U+W)\ne0$, where $U$ is the real finite-range interaction and $W$ a small complex perturbation. For the family above, take $U_{\{x,y\}}(\sigma)=-\beta\sigma_x\sigma_y$ on nearest-neighbour edges and single-spin measures $\nu_x(s)=e^{\beta\eta_xs}/(2\cosh(\beta\eta_x))$ for $s\in\{-1,1\}$. The added field is then the perturbation $W_{\{x\}}(\sigma_x)=-\beta h\sigma_x$, with coefficient norm $\beta|h|$. The criterion therefore gives, for each fixed $0<\beta<\beta_c$, a number $r_\beta>0$ such that
\eq{\inf_{\eta,\Lambda}\calL^\eta_{\beta,\Lambda}\ge r_\beta.}
For this application, the real fields are absorbed into the single-spin probability measures. On each fixed-size block, $|e^{\beta hM_\Lambda}-1|\le e^{\beta|h||\Lambda|}-1$, independently of the fields. Thus the finite-size argument is uniform even for unbounded real field configurations. The following conjecture asks for the near-critical scale of this zero-free neighbourhood.
\begin{conj}\label{conj:random-field}
There are $c>0$ and $\beta_0<\beta_c$ such that, for $\beta_0<\beta<\beta_c$ and every real field configuration $\eta$,
\eq{\calL^\eta_\beta\ge c m_\beta^{15/8}.}
Here $m_\beta$ is the mass of the zero-field model.
\end{conj}
Writing $\mathbf E_\eta$ for disorder expectation, \cref{conj:random-field} would imply $\mathbf E_\eta[\calL^\eta_\beta]\ge c m_\beta^{15/8}$ for any field law. A uniform reverse inequality is impossible: for $\eta_x=a\ne0$, one has $Z^\eta_{\beta,\Lambda}(h)=Z_{\beta,\Lambda}(h+a)$ and hence $\calL^\eta_\beta=\sqrt{a^2+\calL_\beta^2}\ge|a|$, whereas $m_\beta\to0$. Averaging the partition function gives a different object: for independent identically distributed symmetric fields with $\mathbf E_\eta e^{\beta|\eta_0|}<\infty$, independence yields $\mathbf E_\eta Z^\eta_{\beta,\Lambda}(h)=(\mathbf E_\eta e^{\beta\eta_0})^{|\Lambda|}Z_{\beta,\Lambda}(h)$, so its zeros are exactly those of the pure model.

For the RFIM itself, we expect a separation of the two gaps below the critical temperature of the pure model. Let the $\eta_x$ be independent centered Gaussian variables of variance $\varepsilon^2>0$, write $\EE^\eta_\beta$ for infinite-volume thermal expectation at $h=0$, and define the mass from the disorder-averaged connected correlation by
\eq{
 m^{\mathrm{RF}}_{\beta,\varepsilon}
 :=\liminf_{n\to\infty}-\frac1n\log\mathbf E_\eta\!\left[
 \EE^\eta_\beta[\sigma_0\sigma_{(n,0)}]
 -\EE^\eta_\beta[\sigma_0]\EE^\eta_\beta[\sigma_{(n,0)}]
 \right].
}
Ding and Xia \cite[Theorem 1.1]{DX} establish the averaged boundary-influence estimate $\mathbf E[\langle\sigma_0\rangle_{\Lambda_N,+}-\langle\sigma_0\rangle_{\Lambda_N,-}]\le C e^{-cN}$ for every finite $\beta>0$ and independent centered Gaussian fields of variance $\varepsilon^2>0$. Here $+$ and $-$ denote boundary conditions, and the constants may depend on $\beta,\varepsilon$. For our field $\eta$, taking $N=\lfloor n/2\rfloor$ bounds the disorder-averaged connected correlation between $0$ and $(n,0)$ and yields $m^{\mathrm{RF}}_{\beta,\varepsilon}>0$. We conjecture that nevertheless
\eq{
 \mathbf E_\eta[\calL^\eta_\beta]=0
 \qquad\text{for every }\beta>\beta_c\text{ and }\varepsilon>0,
}
which, by nonnegativity, is equivalent to $\calL^\eta_\beta=0$ almost surely. The heuristic is that arbitrarily large regions with nearly vanishing random field resemble the ordered pure model and can bring zeros close to the origin; controlling their effect on the full partition function remains part of the conjecture. The localization results in \cite[Sections 6.3 and 7.1]{AW} exhibit the analogous phenomenon for random Schr\"odinger operators: in the localization regime, an eigenfunction may satisfy $H\psi_E=E\psi_E$ and $|\psi_E(x)|\le C_{\psi_E}e^{-c|x-x_E|}$, with $E\in\sigma(H)$ and localization centre $x_E$. Thus spatial decay can coexist with zero distance from the spectrum, paralleling the proposed coexistence of a positive mass and a closed Lee--Yang gap.

\bigskip\bigskip

The proof of \cref{thm:main} treats the upper and lower bounds separately. The lower bound follows readily from the Lee--Yang product and standard planar Ising estimates, and is postponed to \cref{sec:upper}. The correlation inequality of Ding, Song and Sun \cite{DSS} bounds the covariance of two spins in arbitrary real fields by their zero-field correlation on the same ferromagnetic graph. For the upper bound, \cref{sec:trees} applies this comparison to the spins obtained by separately contracting a block and its shell, and combines it with random-current moment bounds to obtain a finite block expansion and a zero-free criterion. In \cref{sec:lower}, planar crossing bounds and near-critical two-point estimates verify this criterion at the correlation length; the Lee--Yang product then yields the bounds on field derivatives. In \cref{sec:cjn}, we explain the obstruction to a shorter argument based on Ursell-function monotonicity and give that argument conditionally.

The exponent $15/8$ is consistent with the continuum scaling description of the planar Ising model. Fonseca and Zamolodchikov \cite{FZ} study the analytic structure of the free-energy scaling function in the variable $h_{\rm cont}/|m_{\rm cont}|^{15/8}$, where $h_{\rm cont}$ is the continuum magnetic coupling and $m_{\rm cont}$ the thermal mass parameter, and investigate the Yang--Lee edge through numerical calculations and analyticity arguments. The leading-order identifications $h_{\rm cont}\propto\beta h$ and $|m_{\rm cont}|\propto\beta_c-\beta$ give the same exponent $15/8$ in the present comparison. For the lattice model, Ott's analyticity theorem \cite{Ott} states that exponential weak mixing at $(\beta,H)$ implies analyticity of the pressure in a neighbourhood of that point, with $H$ the dimensionless magnetic field. At $H=0$ and $\beta<\beta_c$, the required mixing holds; the change of variables $H=\beta h$ therefore gives analyticity near every such $(\beta,0)$ in our convention. Jiang and Newman \cite[Theorem 1 and Corollary 1]{JN} subsequently identify the limiting first zero with the analytic radius: for increasing finite free-boundary domains $V_n$ exhausting $\ZZ^d$, they prove $\alpha_1(V_n,\beta)\downarrow R_\beta^{\rm an}$. Here $\alpha_1$ is the modulus of the first zero in the dimensionless field $H$, and $R_\beta^{\rm an}$ is the radius of the largest disk centred at zero on which the infinite-volume pressure is analytic; this radius is positive precisely when $\beta<\beta_c$. With $d=2$, $V_n=\Lambda_n$ and $H=\beta h$, their quantities become $\alpha_1(\Lambda_n,\beta)=\beta\calL_{\beta,\Lambda_n}$ and $R_\beta^{\rm an}=\beta\calL_\beta$. These rigorous results establish positivity and the thermodynamic interpretation of the gap; we are not aware of the near-critical comparison $\calL_\beta\asymp m_\beta^{15/8}$ appearing earlier in the literature, perhaps because the question becomes natural only in the context of the Combes-Thomas estimate from Schr\"odinger operators.

Constants $c,C$ may change between occurrences, but are independent of the volume, order and $\beta$ sufficiently close to $\beta_c$.

	\bigskip
	\bigskip
		\paragraph{Acknowledgements.}
JS was supported in part by NSF grant DMS-2510207 and the "ChatGPT for Academic Researchers program" of OpenAI. JC was supported in part by \emph{humanize}. The authors wish to thank Michael Aizenman for stimulating discussions. NB and JS previously posed this question in \cite{Borgnia}, without a solution.
	\bigskip

\section{Preliminaries}
Write all positive imaginary zeros as $i\calL_{\beta,\Lambda,j}$, counted with multiplicity, with $\calL_{\beta,\Lambda,1}=\calL_{\beta,\Lambda}\geq\calL_\beta$. Hadamard factorization \cite[Chapter 5, Theorem 5.1]{SS03} of the even entire function $Z_{\beta,\Lambda}$, of order one, gives the paired Lee--Yang product
\eql{\label{eq:product}
\EE_{\beta,\Lambda}\left[\exp\br{\beta h M_\Lambda}\right]=\frac{Z_{\beta,\Lambda}(h)}{Z_{\beta,\Lambda}(0)}
 =\prod_{j\ge1}\left(1+\frac{h^2}{\calL_{\beta,\Lambda,j}^2}\right)\,.
} Hence 
\eq{
\partial_h \log\br{\EE_{\beta,\Lambda}\left[\exp\br{\beta h M_\Lambda}\right]} = \beta \EE_{\beta,\Lambda,h}\left[M_\Lambda\right] = \sum_{j\geq 1}\frac{2h}{\calL_{\beta,\Lambda,j}^2+h^2}\,.
}

From this we obtain
\eql{\label{eq:relation between expectation and variance of magnetization}
\beta \EE_{\beta,\Lambda,h}\left[M_\Lambda\right] \geq h \sum_{j\geq 1}\frac{1}{\calL_{\beta,\Lambda,j}^2} = h \frac{\beta^2}{2}\Var_{\beta,\Lambda}\left[M_\Lambda\right]\qquad(h\leq\calL_\beta)\,.
}

\paragraph{Notation}
Throughout, we write
\begin{enumerate}
    \item $M_\Lambda := \sum_{x\in \Lambda}\sigma_x$ for the magnetization in $\Lambda$.
    \item $\chi_\beta\equiv\sum_{x\in\ZZ^2}\EE_\beta[\sigma_0\sigma_x]$ for the infinite-volume susceptibility.
\end{enumerate}

\section{The lower bound}\label{sec:upper}
We begin with a short
\begin{lem}\label{lem:planar}
For $\beta<\beta_c$ sufficiently close to $\beta_c$ and all $h>0$,
\eql{\label{eq:planar-inputs}
 \chi_\beta
 \ge c m_\beta^{-7/4},
 \qquad \EE_{\beta_c,h}[\sigma_0]\le C h^{1/15}.
}
\end{lem}
\begin{proof}
Write $t=\beta_c-\beta$. The susceptibility asymptotic in
\cite[equations (2.43)--(2.45)]{WMTB} is
\eq{
 T\chi(T)\sim C_{0,+}(1-T_c/T)^{-7/4}
 \qquad(T\downarrow T_c),
}
where $T$ is the temperature, $\chi(T)$ is the zero-field
susceptibility with respect to the physical magnetic field, and
$C_{0,+}>0$. Taking their couplings $E_1=E_2=1$ and
$T=\beta^{-1}$, we have $T_c=\beta_c^{-1}$ and
$1-T_c/T=t/\beta_c$. Their susceptibility therefore gives
\eq{
 \chi_\beta
 =T\chi(T)\sim C_{0,+}\beta_c^{7/4}t^{-7/4}.
}
Since $m_\beta\sim4t$ by \cref{eq:rigorous mass gap}, this proves
the first estimate.

For the second estimate, \cite[Theorem 1.1]{CGN} bounds the
infinite-volume critical magnetization by $C H^{1/15}$ for small
$H>0$, where $H$ is the coefficient of $\sum_x\sigma_x$ in the
exponent of the Gibbs weight. Setting $H=\beta_c h$ gives
$\EE_{\beta_c,h}[\sigma_0]\le C h^{1/15}$ for small $h>0$,
after adjusting $C$. Since the magnetization is at most one,
increasing $C$ extends this bound to all $h>0$.
\end{proof}

With this, we are now ready for the
\begin{proof}[Proof of the lower bound in \cref{eq:main}]
Griffiths' inequalities and the summable bound \cref{eq:decay} give
\eql{\label{eq:variance-limit}
 \lim_{L\to\infty}\frac{\Var_{\beta,\Lambda_L}\left[M_L\right]}{|\Lambda_L|}
 =\chi_\beta.
}
For the lower limit, restrict to bounded displacements and sites far from the boundary before letting these cutoffs grow; the reverse bound follows from infinite-volume domination. For $0<h\le\calL_\beta$, \cref{eq:relation between expectation and variance of magnetization}, gives
\eql{\label{eq:zero-compare}
 \frac{\beta h}{2}\Var_{\beta,\Lambda_L}\left[M_L\right]
 \le\EE_{\beta,\Lambda_L,h}[M_L]
 \le|\Lambda_L|\EE_{\beta_c,h}[\sigma_0],
}
where the last inequality is monotonicity. Divide by $|\Lambda_L|$, let $L\to\infty$, and use \cref{lem:planar} at $h=\calL_\beta>0$:
\eql{\label{eq:key-upper}
 c\calL_\beta m_\beta^{-7/4}
 \le\frac{\beta\calL_\beta}{2}\chi_\beta
 \le\EE_{\beta_c,\calL_\beta}[\sigma_0]
 \le C\calL_\beta^{1/15}.
}
Hence $\calL_\beta^{14/15}\le C m_\beta^{7/4}$, as required.
\end{proof}

\section{A block expansion}\label{sec:trees}

We establish a criterion for a zero-free region in terms of the
susceptibility and the interaction between a block and its surrounding
shell. The proof requires two estimates: a Gaussian bound on block
magnetization and a comparison of block distributions under different
surrounding spin configurations. We first prove these estimates and
then combine them in a finite block expansion.

\begin{lem}[Gaussian moment bound]\label{lem:gaussian-moments}
For every finite $\Lambda\subset\ZZ^2$ and every $B\subset\Lambda$,
the zero-field magnetization satisfies
\eql{\label{eq:gauss}
 \EE_{\beta,\Lambda}[e^{tM_B}]
 \le \ee^{\frac12t^2\EE_{\beta,\Lambda}[M_B^2]}
 \le \ee^{\frac12t^2|B|\chi_\beta}\qquad(t\in\mathbb R).
}
\end{lem}
\begin{proof}
Set $v_B:=\EE_{\beta,\Lambda}[M_B^2]$. At zero field with free
boundary conditions, spin-flip symmetry implies that all odd moments
of $M_B$ vanish. Since $|M_B|\le |B|$, we may average the exponential
series term by term to obtain
\eq{
 \EE_{\beta,\Lambda}[e^{tM_B}]
 =\sum_{k=0}^{\infty}\frac{t^{2k}}{(2k)!}
   \EE_{\beta,\Lambda}[M_B^{2k}].
}
To bound this series by $\ee^{t^2v_B/2}$, we will show that
\eq{
 \EE_{\beta,\Lambda}[M_B^{2k}]
 \le \frac{(2k)!}{2^k k!}\,v_B^k
 \qquad(k\ge0).
}

The correlation inequality in \cite[Eq.~(12.4)]{Aiz} applies to
finite zero-field ferromagnetic Ising models. In our model on
$\Lambda$, it gives
\eq{
 \EE_{\beta,\Lambda}\!\left[\prod_{i=1}^{2k}\sigma_{x_i}\right]
 \le
 \sum_{j=2}^{2k}
 \EE_{\beta,\Lambda}[\sigma_{x_1}\sigma_{x_j}]
 \EE_{\beta,\Lambda}\!\left[
   \prod_{i\ne1,j}\sigma_{x_i}\right]
}
for $k\ge1$ and $x_1,\ldots,x_{2k}\in\Lambda$, including repeated
sites. We apply this inequality to the expansion
\eq{
 \EE_{\beta,\Lambda}[M_B^{2k}]
 =
 \sum_{x_1,\ldots,x_{2k}\in B}
 \EE_{\beta,\Lambda}\!\left[
   \prod_{i=1}^{2k}\sigma_{x_i}\right].
}
For each fixed $j$, the resulting sum factors as
\eq{
 \sum_{x_1,\ldots,x_{2k}\in B}
 \EE_{\beta,\Lambda}[\sigma_{x_1}\sigma_{x_j}]
 \EE_{\beta,\Lambda}\!\left[
   \prod_{i\ne1,j}\sigma_{x_i}\right]
 =
 \left(\sum_{x,y\in B}
   \EE_{\beta,\Lambda}[\sigma_x\sigma_y]\right)
 \EE_{\beta,\Lambda}[M_B^{2k-2}]=
 v_B\,\EE_{\beta,\Lambda}[M_B^{2k-2}]\,.
}
Indeed, the first factor depends only on the indices $x_1,x_j$,
and the second on the remaining indices. There are $2k-1$ choices
of $j$, so
\eq{
 \EE_{\beta,\Lambda}[M_B^{2k}]
 \le (2k-1)v_B\,\EE_{\beta,\Lambda}[M_B^{2k-2}].
}
Iterating this inequality gives, for $k\ge1$,
\eq{
 \EE_{\beta,\Lambda}[M_B^{2k}]
 \le (2k-1)(2k-3)\cdots1\,v_B^k
 =\frac{(2k)!}{2^k k!}\,v_B^k,
}
and the claimed bound is an equality for $k=0$. Substituting these
moment bounds into the exponential series, whose even coefficients
are nonnegative for every real $t$, yields
\eq{
 \begin{aligned}
 \EE_{\beta,\Lambda}[e^{tM_B}]
 \le
 \sum_{k=0}^{\infty}
 \frac{t^{2k}}{(2k)!}\frac{(2k)!}{2^k k!}\,v_B^k
 =\sum_{k=0}^{\infty}\frac{(t^2v_B/2)^k}{k!}
 =\ee^{t^2v_B/2}.
 \end{aligned}
}
This proves the first inequality in \cref{eq:gauss}.

For the second inequality, the Griffiths--Kelly--Sherman
inequalities \cite{KS} give
\eq{
 \langle\sigma_x\sigma_y\sigma_u\sigma_v\rangle
 -\langle\sigma_x\sigma_y\rangle
  \langle\sigma_u\sigma_v\rangle
 \ge0.
}
This difference is the derivative of
$\langle\sigma_x\sigma_y\rangle$ with respect to the coefficient of
$\sigma_u\sigma_v$ in the Gibbs exponent. Thus two-point correlations
increase with each ferromagnetic coupling. Embedding $\Lambda$ in
larger free boxes, increasing the missing couplings from zero to
$\beta$, and passing to infinite volume gives
\eq{
 \EE_{\beta,\Lambda}[\sigma_x\sigma_y]
 \le \EE_\beta[\sigma_x\sigma_y]
 \qquad(x,y\in\Lambda).
}
Consequently,
\eq{
 \begin{aligned}
 v_B
 =\sum_{x,y\in B}
   \EE_{\beta,\Lambda}[\sigma_x\sigma_y]
 \le\sum_{x\in B}\sum_{y\in B}
   \EE_\beta[\sigma_x\sigma_y]
 \le\sum_{x\in B}\sum_{y\in\ZZ^2}
   \EE_\beta[\sigma_x\sigma_y]
 =|B|\chi_\beta.
 \end{aligned}
}
Here the second inequality uses nonnegativity of the two-point
correlations, and the last equality uses translation invariance.
Substituting this bound for $v_B$ proves the second inequality
in \cref{eq:gauss}.
\end{proof}
The moment bound will control the factors $e^{\beta hM_B}-1$
introduced by the magnetic field. To control their dependence across
blocks, we introduce a correlation between a block and its surrounding
shell.

We now allow the magnetic field to vary from site to site.
For a finite box $\Lambda\subset\ZZ^2$, let $E(\Lambda)$ be its
nearest-neighbour edges. Given
$\mathbf h=(h_x)_{x\in\Lambda}\in\mathbb R^\Lambda$,
let $\mathbb{P}_{\beta,\Lambda,\mathbf h}$ denote the Gibbs
probability measure with free boundary conditions and weights
proportional to
\eq{
 \exp\left(
 \beta\sum_{\{x,y\}\in E(\Lambda)}\sigma_x\sigma_y
 +\beta\sum_{x\in\Lambda}h_x\sigma_x
 \right).
}
When $h_x=h$ for every $x$, this is the previously defined
uniform-field model.

For disjoint nonempty sets $S,T\subset\Lambda$, define
\eq{
 A_{S,T}
 :=
 \{\text{all spins in }S\text{ agree}\}
 \cap
 \{\text{all spins in }T\text{ agree}\},
}
and let
\eq{
 q_{\beta,\Lambda}(S,T)
 :=
 \sum_{s,t\in\{-1,1\}}st\,
 \mathbb{P}_{\beta,\Lambda,\mathbf 0}
 \left[
 \sigma_S\equiv s,\,
 \sigma_T\equiv t
 \,\middle|\,A_{S,T}
 \right].
}
Thus $q_{\beta,\Lambda}(S,T)$ is the zero-field correlation of
the two spins obtained by contracting $S$ and $T$ separately.

\begin{lem}\label{lem:ratio}
For every finite box $\Lambda\subset\ZZ^2$, every $\beta>0$,
every $\mathbf h\in\mathbb R^\Lambda$, every pair of disjoint
nonempty sets $S,T\subset\Lambda$, and all configurations
$\tau\in\{-1,1\}^{S}$ and $\eta\in\{-1,1\}^{T}$,
\eql{\label{eq:ratio}
 \mathbb{P}_{\beta,\Lambda,\mathbf h}
   \left[\sigma_T=\eta
   \,\middle|\,\sigma_S=\tau\right]
 \ge
 \bigl(1-4q_{\beta,\Lambda}(S,T)\bigr)
 \mathbb{P}_{\beta,\Lambda,\mathbf h}
   \left[\sigma_T=\eta
   \,\middle|\,\sigma_S\equiv+1\right].
}
Here $\sigma_S\equiv+1$ means that every spin in $S$ is fixed
to $+1$.
\end{lem}

\begin{proof}
The proof is an application of the FKG and DSS inequalities.

Fix $\beta,\Lambda,\mathbf h,S,T$ and $\tau\in\{\pm1\}^S$. All the conditional probabilities
below are strictly positive. The FKG lattice condition implies
that
\eq{
 \{\pm1\}^T\ni\eta\mapsto\frac{
 \mathbb{P}_{\beta,\Lambda,\mathbf h}
   \left[\sigma_T=\eta
   \,\middle|\,\sigma_S=\tau\right]
 }{
 \mathbb{P}_{\beta,\Lambda,\mathbf h}
   \left[\sigma_T=\eta
   \,\middle|\,\sigma_S\equiv+1\right]
 }
}
is decreasing. Its minimum is therefore
attained when $\eta\equiv+1$. Moreover, the probability that
all spins in $T$ equal $+1$ is increasing in the imposed
configuration on $S$. Consequently,
\eq{
 \frac{
 \mathbb{P}_{\beta,\Lambda,\mathbf h}
   \left[\sigma_T=\eta
   \,\middle|\,\sigma_S=\tau\right]
 }{
 \mathbb{P}_{\beta,\Lambda,\mathbf h}
   \left[\sigma_T=\eta
   \,\middle|\,\sigma_S\equiv+1\right]
 }
 \ge
 \frac{
 \mathbb{P}_{\beta,\Lambda,\mathbf h}
   \left[\sigma_T\equiv+1
   \,\middle|\,\sigma_S=\tau\right]
 }{
 \mathbb{P}_{\beta,\Lambda,\mathbf h}
   \left[\sigma_T\equiv+1
   \,\middle|\,\sigma_S\equiv+1\right]
 }
 \ge
 \frac{
 \mathbb{P}_{\beta,\Lambda,\mathbf h}
   \left[\sigma_T\equiv+1
   \,\middle|\,\sigma_S\equiv-1\right]
 }{
 \mathbb{P}_{\beta,\Lambda,\mathbf h}
   \left[\sigma_T\equiv+1
   \,\middle|\,\sigma_S\equiv+1\right]
 }.
}

To bound the last ratio, condition on $A_{S,T}$, which
contracts $S$ and $T$ separately. After summing over the
spins in $\Lambda\setminus(S\cup T)$, the conditional
distribution of their two common spins has the form
\eq{
 \begin{aligned}
 &\mathbb{P}_{\beta,\Lambda,\mathbf h}
 \left[
 \sigma_S\equiv s,\,
 \sigma_T\equiv t
 \,\middle|\,A_{S,T}
 \right]
 \\
 &\qquad\propto
 \exp\left(
 K_{\beta,\Lambda,\mathbf h}(S,T)st
 +H^S_{\beta,\Lambda,\mathbf h}(S,T)s
 +H^T_{\beta,\Lambda,\mathbf h}(S,T)t
 \right),
 \qquad s,t\in\{-1,1\}.
 \end{aligned}
}
Here the proportionality constant is independent of $s,t$,
and $K_{\beta,\Lambda,\mathbf h}(S,T)\ge0$ by ferromagnetism.
The normalizing factors cancel in the following product of
ratios:
\eq{
 e^{4K_{\beta,\Lambda,\mathbf h}(S,T)}
 =
 \frac{
 \mathbb{P}_{\beta,\Lambda,\mathbf h}
   \left[\sigma_T\equiv+1
   \,\middle|\,\sigma_S\equiv+1\right]
 }{
 \mathbb{P}_{\beta,\Lambda,\mathbf h}
   \left[\sigma_T\equiv+1
   \,\middle|\,\sigma_S\equiv-1\right]
 }
 \times
 \frac{
 \mathbb{P}_{\beta,\Lambda,\mathbf h}
   \left[\sigma_T\equiv-1
   \,\middle|\,\sigma_S\equiv-1\right]
 }{
 \mathbb{P}_{\beta,\Lambda,\mathbf h}
   \left[\sigma_T\equiv-1
   \,\middle|\,\sigma_S\equiv+1\right]
 }.
}
The second factor is at least one, again by FKG monotonicity.
Combining this observation with the preceding lower bound gives
\eq{
 \mathbb{P}_{\beta,\Lambda,\mathbf h}
   \left[\sigma_T=\eta
   \,\middle|\,\sigma_S=\tau\right]\ge
 e^{-4K_{\beta,\Lambda,\mathbf h}(S,T)}
 \mathbb{P}_{\beta,\Lambda,\mathbf h}
   \left[\sigma_T=\eta
   \,\middle|\,\sigma_S\equiv+1\right]\,.
}

It remains to bound $K_{\beta,\Lambda,\mathbf h}(S,T)$ by
$q_{\beta,\Lambda}(S,T)$. Define a new field vector
$\widetilde{\mathbf h}$ by
\eq{
 \widetilde h_x
 =
 h_x
 -\frac{H^S_{\beta,\Lambda,\mathbf h}(S,T)}
        {\beta|S|}\mathbf 1_{\{x\in S\}}
 -\frac{H^T_{\beta,\Lambda,\mathbf h}(S,T)}
        {\beta|T|}\mathbf 1_{\{x\in T\}},
 \qquad x\in\Lambda.
}
Conditionally on $A_{S,T}$, this change subtracts precisely the
two linear terms from the effective exponent. Thus, under
$\mathbb{P}_{\beta,\Lambda,\widetilde{\mathbf h}}
 \left[\,\cdot\mid A_{S,T}\right]$,
the distribution of the two common spins is proportional to
$e^{K_{\beta,\Lambda,\mathbf h}(S,T)st}$. Both common spins have
mean zero, and their covariance is
$\tanh K_{\beta,\Lambda,\mathbf h}(S,T)$.

In \cite[Corollary 1.3]{DSS}, Ding, Song and Sun prove that
the covariance of two spins in a finite ferromagnetic Ising
model with arbitrary real fields is bounded above by their
correlation on the same graph with all fields set to zero.
We apply this result on the graph obtained from $\Lambda$
by separately contracting $S$ and $T$, retaining all bonds
with their multiplicities, and with the fields induced by
$\widetilde{\mathbf h}$. In the dimensionless convention of
that result, the fields at the two contracted vertices are
$\beta\sum_{x\in S}\widetilde h_x$ and
$\beta\sum_{x\in T}\widetilde h_x$, and the remaining fields
are $\beta\widetilde h_x$. The comparison therefore gives
\eq{
 \tanh K_{\beta,\Lambda,\mathbf h}(S,T)
 \le
 \sum_{s,t\in\{-1,1\}}st\,
 \mathbb{P}_{\beta,\Lambda,\mathbf 0}
 \left[
 \sigma_S\equiv s,\,
 \sigma_T\equiv t
 \,\middle|\,A_{S,T}
 \right]=q_{\beta,\Lambda}(S,T)\,.
}
It follows that
\eq{
 e^{-4K_{\beta,\Lambda,\mathbf h}(S,T)}
 =
 \left(
 \frac{1-\tanh K_{\beta,\Lambda,\mathbf h}(S,T)}
      {1+\tanh K_{\beta,\Lambda,\mathbf h}(S,T)}
 \right)^2
 \ge
 \left(
 \frac{1-q_{\beta,\Lambda}(S,T)}
      {1+q_{\beta,\Lambda}(S,T)}
 \right)^2.
}
This proves 
\eql{
\mathbb{P}_{\beta,\Lambda,\mathbf h}
   \left[\sigma_T=\eta
   \,\middle|\,\sigma_S=\tau\right]
 \ge \left(
 \frac{1-q_{\beta,\Lambda}(S,T)}
      {1+q_{\beta,\Lambda}(S,T)}
 \right)^2\mathbb{P}_{\beta,\Lambda,\mathbf h}
   \left[\sigma_T=\eta
   \,\middle|\,\sigma_S\equiv+1\right]\,.
}
The lemma follows from
\eq{
 \left(
 \frac{1-q_{\beta,\Lambda}(S,T)}
      {1+q_{\beta,\Lambda}(S,T)}
 \right)^2
 =
 1-\frac{4q_{\beta,\Lambda}(S,T)}
         {(1+q_{\beta,\Lambda}(S,T))^2}
 \ge 1-4q_{\beta,\Lambda}(S,T).
}
\end{proof}

We now specialize to the following sets. For an integer $\ell\ge1$, in $D_\ell:=[-2\ell,3\ell)^2\cap\ZZ^2$ contract
\eq{
 T_\ell=[0,\ell)^2\cap\ZZ^2,\qquad
 S_\ell=\bigl([-2\ell,3\ell)^2\setminus[-\ell,2\ell)^2\bigr)\cap\ZZ^2
}
separately to two vertices. Denote their zero-field correlation by
 $q_\beta(\ell)$; see \Cref{fig:block-contraction}. Concretely, $A_\ell := A_{S_\ell,T_\ell}, q_\beta(\ell) := q_{\beta,D_\ell}(S_\ell,T_\ell)$. 

\begin{figure}[htbp]
  \centering
  \begin{tikzpicture}[
    x=1.05cm, y=1.05cm,
    font=\small,
    boundary/.style={draw=black, line width=0.6pt},
    extension/.style={draw=black, line width=0.35pt},
    dimension/.style={
      draw=black, line width=0.4pt, <->,
      >={Latex[length=1.4mm,width=1mm]}
    }
  ]
    \fill[black!10] (-2,-2) rectangle (3,3);
    \fill[white]    (-1,-1) rectangle (2,2);
    \fill[black!10] (0,0)   rectangle (1,1);

    \draw[boundary] (-2,-2) rectangle (3,3);
    \draw[boundary] (-1,-1) rectangle (2,2);
    \draw[boundary] (0,0)   rectangle (1,1);

    \node[font=\large] at (0.5,0.5) {$T_\ell$};
    \node[font=\large] at (0.5,2.5) {$S_\ell$};
    \node at (0.5,1.5) {buffer};

    \foreach \x in {0,1}
      \draw[extension] (\x,-0.06) -- (\x,-0.50);
    \draw[dimension] (0,-0.38) --
      node[fill=white,inner sep=2pt] {$\ell$} (1,-0.38);

    \foreach \x in {-1,2}
      \draw[extension] (\x,-1.06) -- (\x,-1.50);
    \draw[dimension] (-1,-1.38) --
      node[fill=black!10,inner sep=2pt] {$3\ell$} (2,-1.38);

    \foreach \x in {-2,3}
      \draw[extension] (\x,-2.06) -- (\x,-2.50);
    \draw[dimension] (-2,-2.38) --
      node[fill=white,inner sep=2pt] {$5\ell$} (3,-2.38);
  \end{tikzpicture}
  \caption{Geometry in $D_\ell\equiv[-2\ell,3\ell)^2\cap\mathbb{Z}^2$.
    The shaded regions are $T_\ell=[0,\ell)^2\cap\mathbb{Z}^2$ and
    $S_\ell=\bigl([-2\ell,3\ell)^2\setminus[-\ell,2\ell)^2\bigr)
    \cap\mathbb{Z}^2$.
    Each is contracted to a separate vertex; vertices in the unshaded
    buffer remain distinct. Lattice vertices and edges are omitted.}
  \label{fig:block-contraction}
\end{figure}
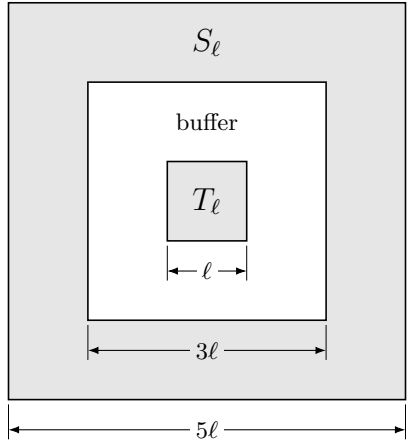


Thus \cref{eq:ratio} compares the conditional distributions of
the spins in $T_\ell$ under arbitrary and all-plus configurations
on $S_\ell$, uniformly in the real field vector $\mathbf h$.

Later, after some spins have been fixed, we apply the same
argument to the remaining domain, with the incident interactions
represented as site-dependent real fields. Although this domain
need not be a box, the proof applies unchanged to any finite
subset of $\ZZ^2$ with its prescribed nearest-neighbour edges.
The resulting zero-field contracted correlation is bounded above
by $q_\beta(\ell)$: restore the deleted vertices and their
incident edges, complete any clipped blocks, and enlarge the
two contracted regions to $S_\ell$ and $T_\ell$. Each of these
operations increases the zero-field correlation by the
Griffiths--Kelly--Sherman inequalities \cite{KS}.

For blocks $B$ of side at most $\ell$, \cref{eq:gauss} gives
\eq{
 \EE_{\beta,\Lambda}\!\left[
   \exp\!\left(t\,\frac{M_B}{\ell\sqrt{\chi_\beta}}\right)
 \right]
 \le \ee^{t^2/2}
 \qquad(t\in\mathbb R).
}
Thus dividing the block magnetization by $\ell\sqrt{\chi_\beta}$
gives a variable with uniformly bounded fluctuations.

We now combine this moment bound
with \cref{lem:ratio} to obtain a zero-free criterion.

\begin{prop}\label{prop:block}
There are absolute $q_*,c>0$ such that, if 
$q_\beta(\ell)\le q_*$, then every finite free square satisfies
\eql{\label{eq:criterion}
 \calL_{\beta,\Lambda}\ge\frac{c}{\beta\ell\sqrt{\chi_\beta}}.
}
\end{prop}
The smallness assumption will be verified in \cref{lem:block-scale},
which shows that one may take $\ell=\lceil A/m_\beta\rceil$, with
$A$ sufficiently large independently of $\beta$. Thus the required
block size diverges as $\beta\uparrow\beta_c$.

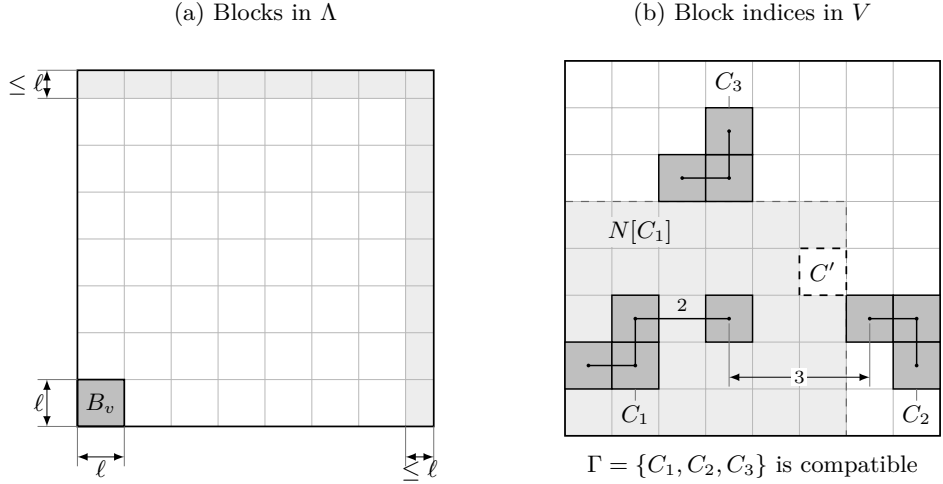
\begin{figure}[htbp]
  \centering
  \begin{tikzpicture}[
    x=0.62cm, y=0.62cm, font=\small,
    gridline/.style={draw=black!28, line width=0.25pt},
    outline/.style={draw=black, line width=0.65pt},
    polycell/.style={draw=black, fill=black!25, line width=0.55pt},
    graphedge/.style={draw=black, line width=0.55pt},
    guide/.style={draw=black!55, line width=0.3pt},
    dimension/.style={draw=black, line width=0.35pt, <->, >=latex}
  ]
    \node[anchor=south] at (3.8,8.55) {(a) Blocks in $\Lambda$};
    \node[anchor=south] at (14.4,8.55) {(b) Block indices in $V$};

    \begin{scope}[shift={(0,0.2)}]
      \fill[black!7] (7,0) rectangle (7.6,7.6);
      \fill[black!7] (0,7) rectangle (7.6,7.6);
      \fill[black!25] (0,0) rectangle (1,1);

      \foreach \k in {1,...,7} {
        \draw[gridline] (\k,0) -- (\k,7.6);
        \draw[gridline] (0,\k) -- (7.6,\k);
      }
      \draw[outline] (0,0) rectangle (7.6,7.6);
      \draw[outline] (0,0) rectangle (1,1);
      \node at (0.5,0.5) {$B_v$};

      \foreach \x in {0,1}
        \draw[guide] (\x,0) -- (\x,-0.85);
      \draw[dimension] (0,-0.65) --
        node[below,inner sep=1pt] {$\ell$} (1,-0.65);

      \foreach \y in {0,1}
        \draw[guide] (0,\y) -- (-0.85,\y);
      \draw[dimension] (-0.65,0) --
        node[left,inner sep=1pt] {$\ell$} (-0.65,1);

      \foreach \x in {7,7.6}
        \draw[guide] (\x,0) -- (\x,-0.85);
      \draw[dimension] (7,-0.65) --
        node[below,inner sep=1pt] {$\le\ell$} (7.6,-0.65);

      \foreach \y in {7,7.6}
        \draw[guide] (0,\y) -- (-0.85,\y);
      \draw[dimension] (-0.65,7) --
        node[left,inner sep=1pt] {$\le\ell$} (-0.65,7.6);
    \end{scope}

    \begin{scope}[shift={(10.4,0)}]
      \fill[black!7] (0,0) rectangle (6,5);
      \draw[gridline,step=1] (0,0) grid (8,8);
      \draw[draw=black!60,dashed,line width=0.45pt]
        (0,5) -- (6,5) -- (6,0);
      \draw[outline] (0,0) rectangle (8,8);
      \node[fill=black!7,inner sep=1pt]
        at (1.6,4.35) {$N[C_1]$};

      \foreach \x/\y in {
        0/1,1/1,1/2,3/2,
        6/2,7/2,7/1,
        2/5,3/5,3/6
      }
        \draw[polycell] (\x,\y) rectangle ++(1,1);

      \draw[graphedge]
        (0.5,1.5) -- (1.5,1.5) -- (1.5,2.5) -- (3.5,2.5);
      \draw[graphedge]
        (6.5,2.5) -- (7.5,2.5) -- (7.5,1.5);
      \draw[graphedge]
        (2.5,5.5) -- (3.5,5.5) -- (3.5,6.5);

      \foreach \x/\y in {
        0/1,1/1,1/2,3/2,
        6/2,7/2,7/1,
        2/5,3/5,3/6
      }
        \fill (\x+0.5,\y+0.5) circle[radius=0.7pt];

      \node[font=\scriptsize,fill=black!7,inner sep=0.6pt]
        at (2.5,2.8) {$2$};

      \node at (1.5,0.45) {$C_1$};
      \draw[guide] (1.5,0.75) -- (1.5,1);

      \node at (7.5,0.45) {$C_2$};
      \draw[guide] (7.5,0.75) -- (7.5,1);

      \node at (3.5,7.55) {$C_3$};
      \draw[guide] (3.5,7.25) -- (3.5,7);

      \foreach \x in {3.5,6.5}
        \draw[guide] (\x,2.4) -- (\x,1.05);
      \draw[dimension] (3.5,1.25) --
        node[fill=white,inner sep=1pt,font=\scriptsize] {$3$}
        (6.5,1.25);

      \draw[fill=white,draw=black,dashed,line width=0.65pt]
        (5,3) rectangle (6,4);
      \node at (5.5,3.5) {$C'$};

      \node at (4,-0.65)
        {$\Gamma=\{C_1,C_2,C_3\}$ is compatible};
    \end{scope}
  \end{tikzpicture}
  \caption{Blocks and compatible polymers.
    (a) Full blocks have side length $\ell$; blocks in the last row
    and column may be clipped. Lattice sites and bonds are omitted.
    (b) Each cell represents one block index. Dark cells form the
    polymers $C_1,C_2,C_3$ in the graph joining indices $v,w$ when
    $0<\|v-w\|_\infty\le2$; only selected edges are drawn.
    A polymer need not be connected by shared block sides.
    Compatibility means
    $\min_{v\in C_i,\,w\in C_j}\|v-w\|_\infty\ge3$ for $i\ne j$.
    The light region is the closed graph neighbourhood $N[C_1]$,
    which other compatible polymers must avoid.
    The dashed singleton $C'$ is disjoint from $C_1$ but incompatible
    with it. A compatible family need not cover $V$.}
  \label{fig:blocks-polymers}
\end{figure}

The proof combines ideas from recursive perfect sampling and polymer
expansions. In \cite[Section 5]{AJ}, Anand and Jerrum construct a sampler
by identifying a part of the conditional spin distribution common to
all surrounding spin configurations; recursive sampling of the
surrounding spins is needed only for the remaining part. We adapt this
construction to blocks, using \cref{lem:ratio} to supply the common part
and assigning independent stacks of random variables to each block.
The subsequent passage from the sets of blocks visited by the recursion
to a polymer expansion is inspired by Ott
\cite[Sections 2 and 5]{Ott}, who encodes spin dependence through random
sets whose separation yields factorization, and uses the resulting
expansion to prove analyticity of the Ising pressure. Here, bounds on
the size of the visited sets are combined with the Gaussian moment
bound \cref{eq:gauss} to obtain a zero-free neighbourhood on the scale
$(\beta\ell\sqrt{\chi_\beta})^{-1}$.

\begin{proof}[Proof of \cref{prop:block}]
We will show that $Z_{\beta,\Lambda}(h)\ne0$ for
$|h|\le\rho/(\beta\ell\sqrt{\chi_\beta})$, with some absolute $\rho>0$.
Tile $\Lambda$ by blocks $B_v$ of side at most $\ell$, clipping the
last row and column, and let $V\subset\ZZ^2$ be their index set; see \Cref{fig:blocks-polymers}.
Define
\eq{
 s=\ell\sqrt{\chi_\beta},\qquad
 X_v=M_{B_v}/s,\qquad
 z=\beta hs,\qquad
 f_v(z)=e^{zX_v}-1.
}
Then
\eq{
 \frac{Z_{\beta,\Lambda}(\frac{z}{\beta s})}{Z_{\beta,\Lambda}(0)}
 =
 \EE_{\beta,\Lambda}\!\left[\prod_{v\in V}(1+f_v(z))\right].
}
Thus it suffices to prove that this expectation is nonzero for
$|z|\le\rho$, uniformly in $V$.

Join block indices at $\ell^\infty$-distance at most two, obtaining
a graph of maximum degree at most $\Delta=24$. We will construct
weights $W_z(C)$, indexed by nonempty connected sets $C\subseteq V$,
such that
\eql{\label{eq:expansion}
 \frac{Z_{\beta,\Lambda}(\frac{z}{\beta s})}{Z_{\beta,\Lambda}(0)}
 =
 \sum_{\Gamma\ {\rm compatible}}\prod_{C\in\Gamma}W_z(C),
}
where compatibility means that the sets in $\Gamma$ are disjoint
and nonadjacent, and the empty family contributes one; see \Cref{fig:blocks-polymers}. The required
bound on these weights is
\eql{\label{eq:small}
 \sup_{v\in V}\sum_{C\ni v}|W_z(C)|\,2^{(\Delta+1)|C|}
 \le1
 \qquad(|z|\le\rho).
}

We first explain why these two properties imply nonvanishing.
For $U\subseteq V$, let $\Xi(U)$ be the compatible sum in
\cref{eq:expansion} restricted to sets contained in $U$, keeping
the weights fixed. Starting from $\Xi(\varnothing)=1$, we prove
inductively that, for every $v\in U$,
\eq{
 \Xi(U)\ne0,\qquad
 \left|\frac{\Xi(U)}{\Xi(U\setminus\{v\})}-1\right|\le\frac12.
}
Indeed, separating families according to the set containing $v$ gives
\eq{
 \Xi(U)=\Xi(U\setminus\{v\})
 +\sum_{\substack{C\subseteq U\\v\in C}}
    W_z(C)\,\Xi(U\setminus N[C]),
}
where $N[C]$ is the closed neighbourhood of $C$. Since
$|N[C]|\le(\Delta+1)|C|$, the induction hypothesis bounds the ratio
$\Xi(U\setminus N[C])/\Xi(U\setminus\{v\})$ in absolute value by
$2^{(\Delta+1)|C|-1}$: delete the other vertices of $N[C]$ one at
a time, using an inverse ratio of at most two at each deletion.
Consequently,
\eq{
 \left|\frac{\Xi(U)}{\Xi(U\setminus\{v\})}-1\right|
 \le\frac12\sum_{C\ni v}|W_z(C)|\,2^{(\Delta+1)|C|}
 \le\frac12.
}
This proves the induction step. It remains to construct
\cref{eq:expansion} and verify \cref{eq:small}.

The construction uses \cref{lem:ratio} to represent the dependence
between blocks by a local sampling procedure. A history $H$ specifies
the spins in a collection of whole blocks. For an unpinned block
$B_v$, let $S_{v,H}$ be the union of the unpinned blocks at index
distance exactly two from $v$; there are at most $b=16$ such blocks.
Once their spins are fixed, the conditional law in $B_v$ depends
only on pins at distance at most two.

Deleting these pinned vertices turns their incident interactions
into real fields. Apply \cref{lem:ratio} in the resulting local domain.
Restoring deleted vertices and edges, completing clipped blocks, and
enlarging the contracted regions to the full block and shell increases
the zero-field contracted correlation by GKS. It is therefore at most
$q_\beta(\ell)$. Choose $0<\delta<1/b$, to be fixed below, and take
$q_*\le\delta/4$. For every shell configuration $\tau$, we obtain
\eql{\label{eq:mixture}
 \begin{aligned}
 &\mathbb{P}_{\beta,\Lambda}
   \left[\sigma_{B_v}=\eta
   \,\middle|\,H,\sigma_{S_{v,H}}=\tau\right]
 \\
 &\qquad=
 (1-\delta)\mathbb{P}_{\beta,\Lambda}
   \left[\sigma_{B_v}=\eta
   \,\middle|\,H,\sigma_{S_{v,H}}\equiv+1\right]
 +\delta R_{\beta,\Lambda,v,H,\tau}(\eta),
 \end{aligned}
}
where $R_{\beta,\Lambda,v,H,\tau}$ is a probability distribution.
Indeed, subtracting the first term leaves a nonnegative measure of
total mass $\delta$.

The common term in \cref{eq:mixture} allows us to sample the target
without discovering its shell with probability $1-\delta$.
With the remaining probability $\delta$, recursively sample the
unpinned shell blocks in a fixed order, retaining earlier shell
outputs, and then sample the target from
$R_{\beta,\Lambda,v,H,\tau}$. An empty unpinned shell requires no
recursion. Each call retains its input history and its direct
children's outputs; temporary descendant pins are discarded when
a child returns.

Use independent stacks of coins and uniform random variables at
each block, reserving the uniform for a call's final draw before
recursing. Each call branches with probability $\delta$ and has
at most $b$ children, so $b\delta<1$ ensures almost sure termination.
To verify exactness, truncate at depth $k$ and use exact conditional
Gibbs draws there. Applying \cref{eq:mixture} backwards gives the
correct marginal at the roots. From $j$ roots, the probability of
reaching depth $k$ is at most $j(b\delta)^k$; letting $k\to\infty$
therefore proves exactness of the untruncated procedure.

For $I\subseteq V$, run the roots in a fixed order, retaining earlier
root outputs. Let $S_v^I$ be the resulting configuration in $B_v$,
and let $Q_I$ be the set of all visited block indices, including
temporary descendants. The root output has the joint Gibbs marginal.
Moreover, the connected components of $Q_I$ factor over their
independent stacks. Deleting roots in other components changes
neither a local history nor a consumed stack entry, since distinct
components are nonadjacent. Conversely, runs on nonadjacent supports
assemble into the joint run. The event of having a specified support,
together with its output, is measurable on that support: simulate
until termination or the first proposed exit.

Write $\EE_{\rm st}$ and $\mathbb{P}_{\rm st}$ for expectation and
probability over the stacks. Every component of $Q_I$ contains a root.
For nonempty connected $C\subseteq V$, define
\eq{
 W_z(C)=
 \sum_{\varnothing\ne I\subseteq C}
 \EE_{\rm st}\!\left[
   \Id_{\{Q_I=C\}}\prod_{v\in I}f_v(z;S_v^I)
 \right],
}
where $f_v(z;S_v^I)$ denotes evaluation at the sampled configuration.
Expanding $\prod_v(1+f_v)$, using exactness of the root marginal,
and grouping the connected components of $Q_I$ proves
\cref{eq:expansion}.

We now bound these weights. Cauchy--Schwarz separates each summand
into the probability of its explored support and a moment of its
field factors. For the latter, positivity of even spin correlations
and \cref{eq:gauss}, applied to $\bigcup_{v\in I}B_v$, give
\eq{
 \EE_{\beta,\Lambda}\!\left[e^{t\sum_{v\in I}|X_v|}\right]
 \le
 \sum_{\varepsilon\in\{-1,1\}^I}
 \EE_{\beta,\Lambda}\!\left[
   e^{t\sum_{v\in I}\varepsilon_vX_v}
 \right]
 \le 2^{|I|}e^{t^2|I|/2},
 \qquad t\ge0.
}
For $|z|\le1$ and $x\in\mathbb R$,
$|e^{zx}-1|^2\le |z|^2e^{4|x|}$. Taking $t=4$ above therefore yields
an absolute $K$ such that
\eql{\label{eq:moment}
 \left(
 \EE_{\beta,\Lambda}\!\left[\prod_{v\in I}|f_v(z)|^2\right]
 \right)^{1/2}
 \le(K|z|)^{|I|},
 \qquad |z|\le1.
}

For the explored support, if at least $N$ calls occur, at least
$(N-|I|)/b$ of their first $N$ coins must branch. There are at most
$2^N$ coin patterns, and the number of distinct visited blocks is
at most the number of calls. Thus
\eql{\label{eq:tail}
 \mathbb{P}_{\rm st}\left[|Q_I|\ge N\right]
 \le 2^N\delta^{(N-|I|)/b}
 =e^{-\lambda N+\eta|I|},
}
where $\eta=b^{-1}\log(1/\delta)$ and $\lambda=\eta-\log2$.
By Cauchy--Schwarz, \cref{eq:moment,eq:tail}, and exactness of the
root marginal,
\eq{
 \left|
 \EE_{\rm st}\!\left[
   \Id_{\{Q_I=C\}}\prod_{v\in I}f_v(z;S_v^I)
 \right]
 \right|
 \le e^{-\lambda|C|/2}
      (Ke^{\eta/2}|z|)^{|I|}.
}
Summing over the nonempty subsets $I\subseteq C$ gives
\eql{\label{eq:activities}
 |W_z(C)|
 \le e^{-\lambda|C|/2}
 \left[(1+Ke^{\eta/2}|z|)^{|C|}-1\right].
}

There are at most $\Delta^{2r-2}$ connected $r$-sets containing
a fixed block, by encoding a spanning tree through its depth-first
walk. Choose $\delta$ sufficiently small that
\eq{
 \Delta^2\,2^{\Delta+1}e^{-\lambda/2}<1.
}
With this $\delta$ fixed, \cref{eq:activities} bounds the sum in
\cref{eq:small} by a convergent geometric series for sufficiently
small $|z|$, and that bound tends to zero as $z\to0$. Hence
\cref{eq:small} holds for $|z|\le\rho$, with an absolute
$0<\rho\le1$.

The induction at the beginning of the proof now gives
$\Xi(V)\ne0$. By \cref{eq:expansion} and $z=\beta hs$, this proves
$Z_{\beta,\Lambda}(h)\ne0$ for
$|h|\le\rho/(\beta\ell\sqrt{\chi_\beta})$, and hence
\cref{eq:criterion} with $c=\rho$.
\end{proof}

\section{The upper bound}\label{sec:lower}

We begin with an upper bound to completement \cref{lem:planar}.
\begin{lem}\label{lem:twopoint}
Uniformly in finite $\Lambda\subset\ZZ^2$, $x,y\in\Lambda$ and $\beta<\beta_c$ sufficiently close to $\beta_c$,
\eql{\label{eq:decay}
 0\le\EE_{\beta,\Lambda}[\sigma_x\sigma_y]
 \le C(1+|x-y|_\infty)^{-1/4}e^{-m_\beta|x-y|_\infty/2}.
}
Consequently,
\eql{\label{eq:sums}
 \sup_{x\in\Lambda}\sum_{y\in\Lambda}\EE_{\beta,\Lambda}[\sigma_x\sigma_y]
 \le\chi_\beta\le C m_\beta^{-7/4}.
}
\end{lem}
\begin{proof}
The Griffiths--Kelly--Sherman inequalities \cite{KS} give nonnegativity of spin-product expectations and monotonicity in each ferromagnetic coupling when the fields are nonnegative. At zero field, adding the missing lattice couplings with strength $\beta$ thus gives $0\le\EE_{\beta,\Lambda}[\sigma_x\sigma_y]\le\EE_\beta[\sigma_x\sigma_y]$. Messager--Miracle-Sol\'e monotonicity \cite{MM,Sch} compares a general displacement with an axial one: $\EE_\beta[\sigma_0\sigma_z]\le\EE_\beta[\sigma_0\sigma_{(n,0)}]$ for $n=|z|_\infty$. With $z=y-x$, translation invariance bounds the finite-volume correlation by this row correlation at $n=|x-y|_\infty$. The transfer-matrix decomposition in \cite[Section V.B, formulas (5.20)--(5.22)]{SML} expresses the row correlation as a sum of nonnegative spectral weights times powers of eigenvalue ratios. For the isotropic choice of both couplings equal to $\beta$, its infinite-volume form is $\EE_\beta[\sigma_0\sigma_{(j,0)}]=\int_{m_\beta}^{\infty}e^{-js}\,d\mu_\beta(s)$, where $\mu_\beta$ is a positive measure and $s$ is minus the logarithm of the eigenvalue ratio. Comparing the integrands at $j=n$ and $j=\lfloor n/2\rfloor$ gives
\eql{\label{eq:spectral}
 \EE_\beta[\sigma_0\sigma_{(n,0)}]
 &\le e^{-m_\beta(n-\lfloor n/2\rfloor)}
 \EE_\beta[\sigma_0\sigma_{(\lfloor n/2\rfloor,0)}]
 \le C e^{-m_\beta n/2}(1+n)^{-1/4}.
}
For the last inequality, Wu's critical row asymptotic \cite{Wu} gives $\EE_{\beta_c}[\sigma_0\sigma_{(j,0)}]\sim a_0j^{-1/4}$ as $j\to\infty$, with $a_0>0$. Apply this at $j=\lfloor n/2\rfloor$ and use monotonicity in $\beta$, increasing the constant to cover bounded $j$. Summation over square shells proves \cref{eq:sums}.
\end{proof}

\begin{lem}\label{lem:block-scale}
There are $c,C>0$ such that, for $\beta<\beta_c$ sufficiently close to
$\beta_c$ and $A\ge1$,
\eql{\label{eq:block-scale}
 q_\beta(\lceil A/m_\beta\rceil)\le Ce^{-cA}.
}
\end{lem}

The proof is modeled on the finite-size and coarse-graining argument
of Duminil-Copin and Manolescu
\cite[Section~2.4, particularly the proof of Proposition~2.13]{DCM},
which compares the crossing scale with the exponential correlation
length.
We adapt their method of extracting separated boxes from long
connections and summing over coarse paths to the two separately wired
terminal sets defining $q_\beta(\ell)$.

\begin{proof}[Proof of \cref{lem:block-scale}]
Write $m=m_\beta$. Let $\phi_p$ be the infinite-volume FK--Ising law,
$p=1-e^{-2\beta}$. The critical FK--Ising RSW theorem \cite[Theorem 1.1]{DHN} gives $c_\rho\le\phi^\xi_{p_c,2,R}(\mathcal C_R)\le1-c_\rho$ for rectangles of fixed aspect ratio $\rho$, uniformly in their size and boundary condition $\xi$. Here $\mathcal C_R$ is a crossing between the chosen opposite sides, and the cluster weight is $2$. At $p_c=1-e^{-2\beta_c}$, the dual crossing bounds apply to the fixed-aspect-ratio rectangles forming each annulus. Uniformity in $\xi$ allows conditioning on the exterior, so there is a uniformly positive conditional probability of a separating dual circuit in each of order $\log K$ annuli between $\Lambda_{2r}$ and $\partial\Lambda_{Kr}$. Their separately wired connection probability
is therefore at most $CK^{-c}$, also for $\beta\le\beta_c$ by
monotonicity. Choose a fixed $K$ making this at most $1/10$.
Averaging \cref{lem:ratio} over infinite-volume boundary spins and
using the Edwards--Sokal coupling gives
\eql{\label{eq:wired}
 \phi^{\rm wired}_{\Lambda_{Kr},p}(E)\le2\phi_p(E)
}
for every event depending on edges in $\Lambda_{2r}$: the conditional
law of those edges given their endpoint spins is the same in both
measures.

Write $\mathcal C(a,b)$ for a horizontal crossing of
$[-a,a]\times[-b,b]$. K\"ohler-Schindler and Tassion's crossing theorem \cite[Theorem 1]{KST} applies to every positively associated percolation law invariant under square-lattice symmetries. It gives $\PP(\mathcal C(\rho r,r))\ge\psi_\rho(\PP(\mathcal C(r,\rho r)))$ for $\rho\ge1$, where $\psi_\rho:[0,1]\to[0,1]$ is an increasing homeomorphism depending only on $\rho$. The FK law $\phi_p$ has the required symmetry and association properties, so taking $\PP=\phi_p$ yields $\phi_p(\mathcal C(\rho r,r))\ge\psi_\rho(\phi_p(\mathcal C(r,\rho r)))$. We will use the aspect ratios $\rho=2$ and $\rho=4$ below.
For a small fixed $\epsilon>0$, let
$n=\min\{r\ge1:\phi_p(\mathcal C(r,r))\le\epsilon\}$.
This is finite by subcritical exponential decay. At
$k=\lfloor n/2\rfloor\ge1$, square crossings exceed $\epsilon$ and
$4k$-by-$2k$ horizontal crossings exceed $\psi_2(\epsilon)$.
Chaining overlapping rectangles, adding vertical crossings in their
square overlaps, and attaching their endpoints by finite energy gives
\eq{
 \EE_\beta[\sigma_0\sigma_{(2Nk,0)}]
 \ge b_{\beta,k}[\epsilon\psi_2(\epsilon)]^{2N+2},
 \qquad b_{\beta,k}>0.
}
The prefactor is independent of $N$. Taking exponential rates yields
$n\le C_\epsilon/m$; $n=1$ is covered by $m\le1$ near criticality.

A connection from $\Lambda_{2n}$ to $\partial\Lambda_{4n}$ forces a
short crossing of one of four $2n$-by-$8n$ rectangles. A long crossing
of the rotated rectangle forces a square crossing at scale $n$.
Thus, with $r=2n$, \cref{eq:wired} and the crossing theorem give
\eq{
 \phi^{\rm wired}_{\Lambda_{Kr},p}
 (\Lambda_r\leftrightarrow\partial\Lambda_{2r})
 \le8\psi_4^{-1}(\epsilon)=:a_\epsilon,
 \qquad r\le C_\epsilon/m.
}
Fix an integer $D>2K+4$, and then choose $\epsilon$ so small that
$100a_\epsilon^{1/D^2}<1/2$.
An ordinary open path between the two separately wired terminal sets
in $q_\beta(\ell)$ traverses a distance of order $\ell$.
Trim it between its last exit from the inner terminal's $(K+2)r$
neighborhood and its first entrance into the outer terminal's such
neighborhood. Record cells of an $r$-grid whenever it exits the
$2r$-box about the preceding cell's centre, and erase loops. This gives
$N\ge c\ell/r-C$ records, with at most $100$ successors per record
and $O(1+\ell/r)$ possible starts. Each nonterminal record forces the
last displayed event. Coloring grid indices modulo $D$ selects at
least $N/D^2-1$ records with disjoint $Kr$-containers avoiding both
terminal sets. Conditional on their exteriors, the FK Markov property
bounds each event by $a_\epsilon$, using wired domination. Summing
record sequences gives
\eq{
 q_\beta(\ell)\le C(1+\ell/r)
       \sum_{N\ge c\ell/r-C}100^N a_\epsilon^{N/D^2-1}
 \le Ce^{-c\ell/r}.
}
Increasing $C$ covers bounded $\ell/r$. This proves \cref{eq:block-scale} at $\ell=\lceil A/m\rceil$.

\end{proof}

We are finally ready for the
\begin{proof}[Proof of the upper and derivative bounds]
Choose $A$ so large that $Ce^{-cA}\le q_*$ in \cref{lem:block-scale}, and set
$\ell=\lceil A/m_\beta\rceil$. Then \cref{prop:block,eq:sums} give, uniformly
in finite free squares,
\eql{\label{eq:finite-gap}
 \calL_{\beta,\Lambda}\ge\frac{c}{\beta\ell\sqrt{\chi_\beta}}
 \ge c m_\beta^{15/8}.
}
Taking the infimum over $\Lambda_L$ proves the upper bound in \cref{eq:main}. 

Next, by \cref{eq:product}, for every $k\ge1$,
\eql{\label{eq:derivative-bound}
 \frac{|\partial_h^{2k}\log Z_{\beta,\Lambda}(0)|}{(2k-1)!}
 =2\sum_{j\ge1}\calL_{\beta,\Lambda,j}^{-2k}
 \le \calL_{\beta,\Lambda}^{2-2k}\partial_h^2\log Z_{\beta,\Lambda}(0).
}
Since $\partial_h^2\log Z_{\beta,\Lambda}(0)=\beta^2\Var_{\beta,\Lambda}(M_\Lambda)\le C|\Lambda|m_\beta^{-7/4}$ by \cref{eq:sums}, the finite-volume gap bound yields \cref{eq:cumulant-main}: the power of $m_\beta$ is $-7/4+(15/8)(2-2k)=2-15k/4$.
\end{proof}

\appendix

\section{The CJN auxiliary inequality and a conditional shorter proof}\label{sec:cjn}
The Ursell-function monotonicity asserted in \cite[Theorem 1, formula (7)]{CJN} is $(-1)^{k-1}\partial_{J_{uv}}u_{2k}\ge0$ for zero-field ferromagnetic Ising models, including repeated marked sites. Applied to the spins and couplings below, it would give a shorter proof of the upper bound. Its proof uses the constrained sign inequality \cite[Proposition 2, formula (43)]{CJN}: a signed partition sum $R$, defined below, is asserted to satisfy $(-1)^{k-1}R\ge0$ whenever no restricted edge is a self-loop or joins the distinguished vertices. As far as we can tell, the following construction at $k=2$ is a counterexample to this auxiliary assertion. The corresponding unconstrained statement \cite[formula (42)]{CJN} is $(-1)^{k-1}R(\mathcal G)\ge0$, with no restrictions on the edge classes. The counterexample concerns the constrained sum and does not disprove either the unconstrained inequality or the monotonicity theorem.

\paragraph{The constrained sign inequality.}
The current-graph partition sum of \cite[Section 3.1]{CJN} is $R(\mathcal G)$, obtained by assigning weight $(-1)^{n-1}(n-1)!$ to each admissible partition into $k+1$ edge classes whose marked vertices form $n$ nonempty even sets. Specializing to $k=2$ and writing $(u,v)$ for their distinguished vertices $(u_0,v_0)$ gives the following definitions.
Let $\mathcal J=\{j_1,j_2,j_3,j_4\}$ be the marked vertices, with
$u=j_1$ and $v\notin\mathcal J$. For a current graph $\mathcal G=(V,\mathcal E)$,
write $\partial E'$ for the odd-degree vertices of an edge set $E'$.
An admissible partition consists of
$\mathcal E=E_1\sqcup E_2\sqcup E_3$, a number $n\in\{1,2\}$ and a
partition $\{P_1,\ldots,P_n\}$ of $\mathcal J$ into nonempty even sets,
such that
\eq{
 \partial E_i=P_i\quad(i<n),\qquad
 \partial E_n=P_n\mathbin{\triangle}\{u,v\},\qquad
 \partial E_i=\varnothing\quad(i>n),
}
and $u$ and $v$ are disconnected in $(V,E_n\cup E_{n+1})$.
Each such partition contributes $(-1)^{n-1}(n-1)!$ to $R(\mathcal G)$.
A restriction $\{e,f\}$ requires $e$ and $f$ to belong to different
$E_i$'s. The constrained inequality, formula (43) in \cite{CJN}, asserts that $(-1)^{k-1}R(\mathcal G;\{e_1,f_1\},\ldots,\{e_s,f_s\})\ge0$ when each pair requires different edge classes and no restricted edge is a self-loop or joins the distinguished vertices. For $k=2$ and $(u_0,v_0)=(u,v)$, it therefore predicts a nonpositive restricted sum.

Take two additional vertices $a,b$, distinct from each other and from
$u,v,j_2,j_3,j_4$, and let
\eql{\label{eq:cjn-graph}
 \mathcal E=\{uj_2,uj_3,vj_4,ua,av,ub,bv\}.
}
The underlying graph also contains the distinguished edge $uv$;
its current multiplicity is zero, while every edge in
\cref{eq:cjn-graph} has multiplicity one. Thus
$\partial\mathcal E=\{j_2,j_3,j_4,v\}
=\mathcal J\mathbin{\triangle}\{u,v\}$, as required.
Impose the two restrictions $\{ua,uj_2\}$ and $\{ub,uj_3\}$.
They satisfy all the hypotheses of formula (43).

Since $a,b$ are unmarked vertices of degree two, each of the paths
$uav$ and $ubv$ lies entirely in one $E_i$.
If $n=1$, the three leaf edges belong to $E_1$, and the disconnection
condition forces both paths into $E_3$. This gives the admissible partition
\eql{\label{eq:cjn-partition}
 E_1=\{uj_2,uj_3,vj_4\},\qquad E_2=\varnothing,\qquad
 E_3=\{ua,av,ub,bv\}.
}
If $n=2$, disconnection in $E_2\cup E_3$ forces both paths into $E_1$.
The restrictions then put $uj_2,uj_3$ in $E_2$; parity at $v$ also
puts $vj_4$ in $E_2$. Consequently $\partial E_1=\varnothing$,
contrary to $\partial E_1=P_1\ne\varnothing$.
Thus \cref{eq:cjn-partition} is the only constrained partition, and
\eql{\label{eq:cjn-counterexample}
 R(\mathcal G;\{ua,uj_2\},\{ub,uj_3\})=1,\qquad
 (-1)^{2-1}R(\mathcal G;\{ua,uj_2\},\{ub,uj_3\})=-1.
}
Without the restrictions there are three partitions with $n=2$ and
one with $n=1$, so $R(\mathcal G)=-2$; the unconstrained sign is correct
for this example.

The reduction in \cite{CJN}, formula (55), separates the restricted partition sum according to whether two incident edges belong to the same edge class. Formula (56) then claims the required sign for the same-class term by merging those edges and invoking induction. Here the selected edges are $ua,av$, and this is where the induction fails.
Replacing $ua,av$ by a single edge $uv$, when they belong to the same
subgraph, turns the restriction $\{ua,uj_2\}$ into $\{uv,uj_2\}$.
The reduced restriction contains the distinguished edge and is therefore
excluded from the induction hypothesis. The excluded example in \cite[Remark 3, formulas
(44)--(45)]{CJN} has two parallel edges $e_1,e_2$ between the distinguished vertices, with restrictions $\{e_1,j_1j_2\}$ and $\{e_2,j_1j_3\}$ and signed sum $-1$. Suppressing $a$ and $b$ turns our paths $uav,ubv$ into $e_1,e_2$, respectively; with $j_1=u$, this recovers precisely that example. Subdividing the two edges makes the original restrictions admissible without changing their negative signed sum.

\paragraph{The shorter proof conditional on Ursell monotonicity.}
Consider any finite ferromagnetic Ising graph at zero field, with
couplings $J_{uv}\ge0$. For a set $I$ of labels, whose sites $x_i$ may
coincide, write $u(I)=u_{|I|}(\sigma_{x_i}:i\in I)$ for the joint
Ursell function. Its definition is
\eq{
 u(I)=\left.\left(\prod_{i\in I}\partial_{t_i}\right)
 \log\EE_\Lambda\!\left[e^{\sum_{i\in I}t_i\sigma_{x_i}}\right]
 \right|_{t=0}.
}
The sign theorem of \cite{Shl} states that $(-1)^{r-1}u_{2r}\ge0$ for even Ursell functions of a zero-field ferromagnetic Ising model. With the generating-function convention above and spin list $(\sigma_{x_i}:i\in I)$, this is the sign rule for $u(I)$ when $|I|=2r$. The additional monotonicity asserted in \cite{CJN} says that the signed even Ursell function is nondecreasing in each ferromagnetic coupling, including when marked sites coincide. Written with their $k=r$, marked sites $j_i=x_i$ and differentiated edge $u_0v_0=uv$, this is the assumption
\eql{\label{eq:cjn-conditional}
 (-1)^{r-1}\partial_{J_{uv}}
 u_{2r}(\sigma_{x_1},\ldots,\sigma_{x_{2r}})\ge0
 \qquad(r\ge1).
}
For an even label set $I$ with $|I|\ge4$ and distinct labels $a,b\in I$,
coupling differentiation and the local-field product rule give
\eql{\label{eq:conditional-recursion}
 \partial_{J_{x_ax_b}}u(I\setminus\{a,b\})
 =u(I)+\sum_{\substack{B\subset I:\ a\in B,\ b\notin B\\ |B|\ \mathrm{even}}}
 u(B)u(I\setminus B).
}
If necessary, add the edge $\{x_a,x_b\}$ with coupling zero.
For $|I|=2k$, multiply \cref{eq:conditional-recursion} by $(-1)^{k-2}$.
The sign rule and \cref{eq:cjn-conditional} imply
\eql{\label{eq:conditional-split}
 |u(I)|\le
 \sum_{\substack{B\subset I:\ a\in B,\ b\notin B\\ |B|\ \mathrm{even}}}
 |u(B)|\,|u(I\setminus B)|.
}
When $x_a=x_b$, perturb instead by $t\sigma_{x_a}^2=t$.
The left side of \cref{eq:conditional-recursion} is then zero and
\cref{eq:conditional-split} holds with equality.

Put $G_{ij}=\EE_\Lambda[\sigma_{x_i}\sigma_{x_j}]$.
Induction on $|I|$ gives
\eql{\label{eq:conditional-tree}
 |u(I)|\le
 \left(\max_{\substack{i,j\in I\\i\ne j}}G_{ij}\right)^{1/2}
 \sum_{\substack{T\ \mathrm{tree}\\\mathrm{on}\ I}}
 \prod_{\{i,j\}\in E(T)}G_{ij}^{1/2}.
}
The two-label case is equality. For the induction step, choose $a,b$
attaining the maximum and apply \cref{eq:conditional-split}.
The two inductive prefactors have product at most $G_{ab}$.
Its two square roots supply the prefactor in \cref{eq:conditional-tree}
and an edge $\{a,b\}$ joining the two trees. This construction is
injective: deleting that fixed edge recovers both trees and the block
containing $a$.

For the lattice model, set
\eq{
 Q_{\beta,\Lambda}=\sup_{x\in\Lambda}\sum_{y\in\Lambda}
 \EE_{\beta,\Lambda}[\sigma_x\sigma_y]^{1/2}.
}
Dropping the prefactor in \cref{eq:conditional-tree}, summing successively
over tree leaves and using Cayley's count gives, for even $n\ge2$,
\eql{\label{eq:conditional-sum}
 \frac{|\partial_h^n\log Z_{\beta,\Lambda}(0)|}{n!}
 \le\frac{\beta^n|\Lambda|n^{n-2}Q_{\beta,\Lambda}^{n-1}}{n!}
 \le\frac{|\Lambda|}{Q_{\beta,\Lambda}}
       (e\beta Q_{\beta,\Lambda})^n.
}
Odd derivatives vanish. The logarithmic Taylor series converges for
$|h|<(e\beta Q_{\beta,\Lambda})^{-1}$, and its exponential equals
$Z_{\beta,\Lambda}(h)/Z_{\beta,\Lambda}(0)$ there by the identity theorem.
Consequently, under \cref{eq:cjn-conditional},
\eql{\label{eq:conditional-gap}
 \calL_{\beta,\Lambda}\ge(e\beta Q_{\beta,\Lambda})^{-1}.
}
Finally, summing the square root of \cref{eq:decay} over square shells
gives $Q_{\beta,\Lambda}\le C m_\beta^{-15/8}$.
Thus \cref{eq:conditional-gap} would prove
$\calL_\beta\ge c m_\beta^{15/8}$, and hence the upper bound in
\cref{eq:main}, without the block expansion or planar crossing argument.
The derivative bounds would then follow from \cref{eq:derivative-bound}
exactly as in \cref{sec:lower}.

\printbibliography[title={References}]
\end{document}